\documentclass[conference]{IEEEtran}
\IEEEoverridecommandlockouts

\usepackage{cite}

\ifCLASSINFOpdf
   \usepackage[pdftex]{graphicx}
   \graphicspath{{./Figures/}}
   \DeclareGraphicsExtensions{.pdf,.jpeg,.png}
   \usepackage{dblfloatfix}
\else
\fi

\usepackage{amsmath}
\usepackage{amssymb}
\usepackage{amsthm} 
\DeclareMathOperator{\EX}{\mathbb{E}}
\DeclareMathOperator{\polylog}{polylog}

\usepackage{algorithmic}
\usepackage{algorithm}

\begin{document}
\title{LightSync: Ultra Light Client for PoW Blockchains}
\newtheorem{lem}{Lemma}
\newtheorem{thm}{Theorem}
\author{\IEEEauthorblockN{Niusha Moshrefi\textsuperscript{\textsection}}
\IEEEauthorblockA{School of Engineering\\
University of British Columbia\\
Okanagan Campus, Canada\\
niushamo@mail.ubc.ca}
\and
\IEEEauthorblockN{Mahyar Daneshpajooh\textsuperscript{\textsection}}
\IEEEauthorblockA{School of Engineering\\
University of British Columbia\\
Okanagan Campus, Canada\\
mahyard@mail.ubc.ca}
\footnotemark
\and
\IEEEauthorblockN{Chen Feng}
\IEEEauthorblockA{School of Engineering\\
University of British Columbia\\
Okanagan Campus, Canada\\
chen.feng@ubc.ca}}

\maketitle

\begin{abstract}
 Full nodes in a blockchain network store and verify a copy of the whole blockchain. Unlike full nodes, light clients are low-capacity devices that want to validate certain data on a blockchain. They query the data they want from a full node. If light clients do not verify the data they receive, full nodes might deceive them. SPV, introduced in the Bitcoin paper, is a practical solution to this problem currently used in many PoW blockchains.
In SPV, the resources needed to verify a full node's response grow linearly with the blockchain size, making it inefficient over the long run. Another issue with SPV is that the full nodes do not get compensated for the services they provide.

In this work, we introduce \textsc{LightSync}, a \emph{simple} and \emph{cost-effective} solution for light clients to verify the inclusion of certain data in a PoW blockchain. The resources needed for running \textsc{LightSync} remain \emph{constant} no matter what the size of the blockchain is.  \textsc{LightSync} uses an incentive mechanism that encourages full nodes to participate in the protocol. We perform a thorough analysis of the security of \textsc{LightSync} and discuss the details of deploying it in a real-world environment.
\end{abstract}

\IEEEpeerreviewmaketitle

\section{Introduction}\label{sec:introduction}
    In a blockchain network, full nodes are nodes that maintain a copy of the whole blockchain and relay data in the network. They need to have high storage, bandwidth, and computation power since they need to store all the blocks and verify newly mined blocks to add them to their local chain. For example, currently, Ethereum's \cite{Ethereum} blockchain size is roughly 870 gigabytes \cite{ethsize}. On the other hand, some clients have access to limited resources, but need instant access to certain data on a blockchain. Nowadays, there exist many such clients such as smartphones, IoT devices, and wearable devices who want to get data from a blockchain without storing and verifying the whole blockchain. These low-capacity clients are called \textit{light clients}. In case that a light client wants to verify data from multiple chains, the need for a cost-effective solution becomes more crucial.
    
    Another concern in today's blockchain ecosystem is the interoperability of the different blockchains \cite{interop, Zamyatin2019SoKCA,belchior2021survey}. Recently, the topic of cross-chain applications has attracted wide attention \cite{btcrelay, rainbow, gavzi2019proof}. In order to enable blockchains to embed a client of another blockchain to facilitate cross-chain transactions and cross-chain data transfer, there is a need for an efficient solution for light clients.
    
    One solution for the light clients who need access to the blockchain data is to send their queries to a trusted third party (who maintains the blockchain) and get the required data. However, this approach contradicts the very concept of decentralization in blockchains.
    Thus, a light client needs an efficient solution to obtain information from any full node and a proof to confirm its validity, inclusion in the blockchain, and finalization.
    Two important problems need to be tackled here. First, the cost of the solution in terms of the required storage and the computation power should be minimized. Second, an incentive mechanism for the full nodes to provide the data and the proof should be designed, especially with massive inquiries from light clients.
    
	\subsection{Related Work}\label{sec:related-work}
    The original Bitcoin paper \cite{Bitcoin} has introduced \textit{SPV clients}, who are light clients performing a Simple Payment Verification by storing and verifying all the block headers of a blockchain. This approach is cheaper than storing all the blocks of the blockchain and acting as a full node. However, it is still costly to store and verify all the block headers. The cost of this solution grows linearly with the size of the blockchain. 
    
    Recently, two sublinear solutions have been proposed to reduce the cost of light clients further. NIPoPoW \cite{NiPoPoWs} and FlyClient \cite{FlyClient} use several randomly chosen blocks from the entire blockchain to verify the validity of a predicate based on that blockchain. 
    Compared to SPV, they have a lower cost for verification; the user is required to download a polylogarithmic fraction of blocks to store and verify. 
	
	NIPoPoW chooses the high-difficulty blocks, as they happen less frequently than other blocks, to be verified as a sub-chain of the underlying chain. Verifying the validity of those blocks proves that enough Proof-of-Work (PoW) for the underlying chain has been provided with a high probability. FlyClient has an optimized approach in which a binary search and a random sampling are conducted to find the invalid blocks of an invalid chain. This way, the algorithm detects and discards invalid proofs. 
	
	Both of the solutions use velvet forks \cite{velvetfork} to include additional structures in the block headers of the source blockchain. NIPoPoW requires an interlink pointers data structure to be included in the block headers and FlyClient requires all block headers to include the root of a Merkle Mountain Range (MMR) \cite{mmr} commitment of all previous block headers.
	However, both solutions suffer from several drawbacks. NIPoPoW only works well when the difficulty of the blockchain stays constant which is an unrealistic assumption.
	On the other hand, FlyClient has not conducted a comprehensive analysis on the security of its protocol under velvet fork deployment. Their protocol is exposed to chain-sewing attacks under such conditions. Nemoz and Zamyatin \cite{cryptoeprint:2021:782} have introduced this attack and proposed some solutions to fix it.

    \subsection{Our Contribution}\label{sec: our-contribution} 
    We propose \textsc{LightSync}, a new protocol for light clients in PoW blockchains. The storage and computation power needed for the users in \textsc{LightSync} do not grow with the length of the blockchain and remain constant. Thus, \textsc{LightSync} dramatically reduces the resources that light clients need for confirming a predicate from a PoW blockchain. Especially for the light clients who want to confirm data from a large or a fast-growing blockchain, \textsc{LightSync} is a promising solution to reduce the cost. For example, the current length of the Ethereum blockchain is roughly 13 Million blocks. Assume the adversary mining power to the honest mining power to be $\frac{1}{2}$ and the adversary miners to use the same target difficulty as honest miners. In case that the blockchain includes MMR root in the block header structure, to achieve the failure probability less than $2^{-20}$, the expected number of block headers for \textsc{LightSync}'s proof will be $140$. In Table~\ref{tab_size}, we demonstrate the resulted improvement in proof size in comparison to SPV.
    
    \begin{table}[!t]
    \caption{Proof size (in MB) of SPV and LightSync for Ethereum blockchain}
    \label{tab_size}
    \centering
    \begin{tabular}{|c|c|c|c|}
    \hline
    \textbf{Blockchain length} & \textbf{1 M} & \textbf{10 M} & \textbf{100 M}  \\ 
    \hline
    SPV & 508 & 5,080 & 50,800\\ 
    \hline
    LightSync & 0.076 & 0.076 & 0.076\\ 
    \hline
    \end{tabular}
    \end{table}
    
    Our solution is simple with realistic assumptions that make it easily applicable to the current blockchains. Moreover, \textsc{LightSync} employs an incentive mechanism that encourages full nodes to participate in the protocol. We have also analyzed the security of the protocol thoroughly \cite{security} and considered all the deployment issues. Therefore, \textsc{LightSync} is the first incentivized protocol for light clients that has the minimum cost for its users which is provably secure.
    

    
\section{Models and Problem Formulation}\label{sec:model-and-problem-definition}

 \begin{figure*}[!b]
  \includegraphics[scale=0.5]{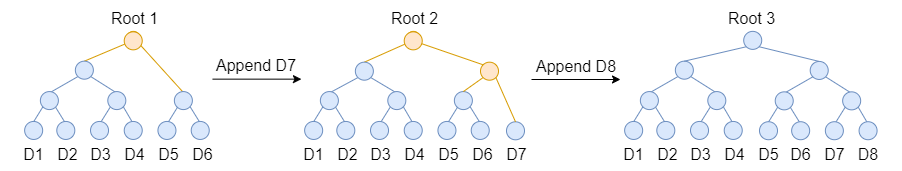}
  \centering
  \caption{Merkle Mountain Range}
  \label{fig1}
 \end{figure*}
 
 \subsection{Backbone Model} \label{sec:backbone-model}
 Our model is based on the standard backbone model for PoW blockchains \cite{10.1007/978-3-319-63688-7_10}. This model consists of three main players: \textit{full nodes}, \textit{miners} and \textit{light clients}. Full nodes maintain a copy of the whole blockchain. They validate every data they receive. Also, some full nodes relay data through the network. Miners create new blocks extending the longest chain of the blockchain and commit new transactions they receive from the clients. Miners compete to become the next block proposer by solving a puzzle which is known as \textit{Proof-of-Work}. For getting blocks and transactions, miners rely on full nodes. Clients need up-to-date information from the blockchain. For example, when they want to perform a transaction inclusion verification, they need the block header of the desired block, in addition to a proof for the transaction inclusion in that block header. Light clients are clients who utilize fewer resources such as storage and computation. They send queries to some full nodes to obtain their desired data instead of paying for the cost of maintaining the whole blockchain.
 
 Block creation is captured in the random oracle model as in \cite{10.1007/978-3-319-63688-7_10}. This model has a security parameter $\sigma$. In each mining round, a miner sends $q$ queries to the random oracle function $H(.)$. For each query, if the value has not been queried before, the function returns a random value from ${\{0,1\}}^\sigma$ and stores the input and output values in a table. If the value has been queried before, the function finds the recorded output from the table and returns it.
 
 For a block to be valid, it needs to have a proof to show enough work has been done on that block. Each block $B$ has a target difficulty $T$. If $H(B) < T$, we say the block meets its difficulty requirement, therefore, it has a sufficient amount of work. The target difficulty $T$ is set in a way that the block intervals remain constant in expectation. We call the expectation of block intervals \textit{Block Interval Time}. The parameter $\lambda$ is the arrival rate of new blocks, which is equal to the inverse of the block interval time. In our model, the target difficulty of the blocks can be constant or variable. Based on the standard backbone model, the active mining power in the network is proportional to $\lambda/T$.

 We abstract $B$ to contain the block header information. A valid block that has a valid PoW and extends the longest chain will be propagated by honest nodes and with a high probability will get finalized in the blockchain. A block gets finalized when it is buried under a chain of at least $k$ other blocks, where $k$ is the \textit{finality parameter}. The finality parameter $k$ is tuned in a way that the probability of a finalized block slipping out of the longest chain is made negligible. The finality parameter differs from one blockchain to another. For example, for the Bitcoin ledger, we know that $k=6$.

 Each block has a block header. Every block header contains a block number, the hash of the previous block as a pointer to it, the target difficulty $T$, the Merkle root of all transactions, and some other information related to the block. The size of block headers is much smaller than the size of the blocks themselves.
 
 When a miner receives a new block, that block may extend some block other than the last one and cause a fork in that miner's local blockchain. Honest nodes will follow the longest chain rule. They choose the longest valid chain and start mining their new block on it. In a valid chain, all the blocks meet their difficulty requirement. We assume more than half of the mining power in the network is honest. Any blockchain that an honest node maintains is called an \textit{honest chain}.
 
 \subsection{Chain Commitment}\label{sec:chain-commitment}
 We leverage the notion of \textit{Merkle Mountain Range (MMR)} in our protocol to enable verifying data from any previous block header of the blockchain. MMR is a more efficient variant of Merkle trees \cite{merkle}. For the MMR root of each block, the MMR leaves are the block headers of all previous blocks. Just like Merkle trees, parent nodes in MMR are hash values of their two children. Using MMR, the entire blockchain can be committed into a single hash value. Also, the proof of the inclusion of a block in the MMR tree is of logarithmic size. MMR uses an efficient updating process to append new leaves and update its root. Each block's MMR root can be constructed from the previous block's MMR root with a low computation overhead. It is also capable of efficiently removing the last leaf from the structure. Moreover, the consistency of the MMR root of a past block with the MMR root of a more recent block can be shown easily. The process of adding a leaf to an MMR tree is shown in Figure \ref{fig1}. Since an MMR tree might have more than one peak, the MMR root is constructed from bagged peaks.
 
 For blockchains that do not have MMR root in their block header structure, an upgrade to the consensus layer is required for the structure to include MMR root in it. One possible way for upgrading is to do a hard fork or a soft fork \cite{hardsoft}. If doing such a fork is not possible, a velvet fork, introduced by Kiayias et al. in \cite{NiPoPoWs}, can be used to add MMR root to the block structure. The changes they propose in the velvet fork need no rule modifications to the consensus layer. Honest miners will be divided into two groups of upgraded and non-upgraded honest miners. The upgraded miners are required to include the MMR root of the previous blocks in their coinbase data, however, non-upgraded miners just ignore that data as comments. Coinbase is the first transaction of a block, determined by the block's miner. Both types of miners accept previously valid blocks regardless of whether they have upgraded structures or not. 
 
 In the velvet fork, the set of accepted blocks is not modified, therefore, there may exist some upgraded blocks containing an invalid MMR root. We assume that the majority of the upgraded miners are honest. We add a voting mechanism so that valid MMR roots can be detected. To enable the voting, we require every upgraded miner to include a buffer of length $\alpha$ in their coinbase data as well as the MMR root. Later, we will define $\alpha$ in a way for the protocol to be secure and discuss it. Every upgraded block has a flag in its coinbase data so that it can be recognized from the non-upgraded blocks.
 
 \subsection{The Prover and Verifier Model}\label{sec:the-prover-and-verifier-model}
 To verify some data being included in the blockchain without having to maintain the whole blockchain, a client needs to send some queries to full nodes. We call the light client a \textit{verifier} and the full nodes \textit{prover}s. We assume that the verifier doesn't trust any of the provers and is attempting to find a prover to get some finalized blocks of the honest chain. This way, she can perform a state or transaction inclusion verification test and verify a piece of data is included in the blockchain. We assume that the verifier knows the block interval time and the confirmation number of the blocks $k$.

 The verifier sends the same request to several provers and they will respond to her with a proof. The verifier needs to determine which of the proofs belongs to an honest prover. The proof contains a chain of blocks. We assume that there exists at least one honest full node among the selected provers, i.e., the client has not been eclipsed from the network.
 
 Without loss of generality, we can assume that the verifier sends requests to two full nodes where one of them is honest and the other one is malicious. The verifier's goal is to determine which of them is the honest node. The honest node maintains the honest chain which is a blockchain with a specific genesis block that has maximum mining power working on it. In other words, it has the highest cumulative difficulty. The verifier wants to confirm some data from this honest chain. In real-world scenarios, the verifier can send requests to multiple provers to increase the probability of communicating with an honest node. 

 \subsection{Notations}\label{sec:notations}
 We denote a chain of length $n$ of block headers by $C[0:n-1]$. In this chain, $C[i]$ ($0\leq i\leq n-1$) refers to the block header with height $i$, and $C[i:j]$ refers to the set of block headers from height $i$ inclusive to height $j$ inclusive. Specifically, $C[i:end]$ is the set of block headers starting from $C[i]$ to the end of the chain.

 \section{LightSync Protocol}\label{sec:lightsync-protocol}
 In this section, we present the protocol of \textsc{LightSync}. The main properties of the \textsc{LightSync} protocol are threefold. First, the proposed protocol is very simple, and it can easily be implemented for any PoW blockchain without the need to change its consensus rule. Second, the computation power and the storage that a verifier needs have been dramatically reduced. The number of the blocks that a verifier needs to download from the full nodes does not depend on the length of the chain and is a constant factor. Third, \textsc{LightSync} uses an incentive mechanism to encourage the full nodes to engage in the protocol.
 
 \subsection{Protocol Overview}\label{sec:protocol-overview}
 As we discussed in the prover and verifier model in the previous section, we assume that we have one verifier and two provers where only one of the provers is honest.

 The verifier starts the protocol by sending a query to the provers. She creates a transaction and sends it to both of the provers. We call this transaction the \textit{query transaction}. We explain the details of this transaction in the incentive section (Section~\ref{sec:incentive-design}).
 After sending requests to the provers, the verifier waits to get responses. Every prover should respond with a proof before a predetermined deadline. The verifier compares the received proofs using \textsc{LightSync} algorithm. The winner proof belongs to the honest node with an overwhelming probability. 
 
 The proof consists of two parts. First, a proof to get the last finalized block. Second, a proof to find the last valid MMR root of the blockchain.
 Using the first part of the proof, the verifier gets the last finalized block of the honest chain. Furthermore, the second part of the proof provides information about the history of the blockchain in case that the verifier needs to verify some data based on the past blocks of the blockchain.
 
 In the following, we explain the details of each part of the proof.
 
 \subsection{Getting the Last Finalized Block Header}\label{part1}
 The verifier sends a \textit{query transaction} to both provers. Provers will try to include the \textit{query transaction} in their local blockchains and provide the verifier with a valid proof within the \textit{challenge period} to prove the honesty of their local blockchains. The \textit{challenge period} is determined by the verifier and is sent to the provers along with the \textit{query transaction}. The longer the \textit{challenge period} is, the higher certainty it provides for the verifier about its final decision. The \textit{challenge period} starts when the verifier sends the \textit{query transaction} to the provers. The steps of the protocol for finding the last finalized block header are described in Algorithm~\ref{alg1}.
 
 To include the \textit{query transaction} in its local blockchain, each prover broadcasts it into his local blockchain network. The honest prover propagates the transaction in the whole blockchain network so that the honest mining power working on the last block of the blockchain includes this transaction in their block.
 After the \textit{query transaction} is included in their local blockchains, the provers will wait for some block confirmations. The provers should send their proofs to the verifier before the \textit{challenge period} is passed. The proof consists of a sub-chain of block headers from the blockchain, as well as a Merkle inclusion proof of the \textit{query transaction} in one of those block headers. The sub-chain contains block headers from the block including the \textit{query transaction} to the last mined block on top of it that gets mined before the \textit{challenge period} is passed. Provers should include at least $k+1$ block headers in their proofs. If the number of blocks mined after the block containing the \textit{query transaction} is less than $k$, the provers should include the previous block headers to contain at least $k+1$ block headers in their proof.
 
 The honest prover can simply wait for the blockchain to grow and have some block confirmations, and send the proof to the verifier a few moments before the deadline. On the other hand, for the malicious prover, the honest mining power will not help him to grow his chain, as long as his local chain data contradicts the honest chain data. The malicious prover, therefore, has no more power than the adversary mining power to mine confirmation blocks on top of the block containing the \textit{query transaction}. We call the confirmation block headers along with the block header including the \textit{query transaction}, the \textit{challenge headers}. Since the adversary mining power is less than the honest mining power due to the underlying blockchain's security, in Section~\ref{sec:security}, we show that the \textit{overall difficulty} of the \textit{challenge headers} provided by the malicious prover in the \textit{challenge period} is less than the honest prover's with a very high probability.\\
 
 \textbf{Definition 1.} \textit{The \textit{overall difficulty} of a group of blocks is the summation of the inverse of the target difficulty of each block.}\\
 
 The verifier receives all the proofs provided within the \textit{challenge period}. Note that both the verifier and the provers should be online so that the verifier can examine whether the proof was provided within the \textit{challenge period} or not. 
 Each proof consists of a chain of block headers. First, the verifier checks the inclusion of her \textit{query transaction} in one of the block headers of the proof using the Merkle proof of inclusion. Then, she checks the validation of the rest of the block headers. Each block header should refer to its previous block header correctly and have valid PoW regarding its target difficulty. If a proof passes all the above checks, we consider it as a \textit{valid proof}. In the next step, the verifier decides which proof to accept by calculating all valid proofs' \textit{overall difficulties} and comparing them together. The valid proof with the highest \textit{overall difficulty} gets accepted. We call this proof the \textit{winner proof}.
 
 Since the malicious prover wants the verifier to accept his local chain, he tries to maximize the \textit{overall difficulty} of his proof. It means that he tries to mine as many blocks with high difficulties as possible, on top of the block containing the \textit{query transaction}. As long as the \textit{query transaction} gets determined by the verifier, the provers cannot start constructing the proof before they get the \textit{query transaction} from the verifier. Therefore, they only have a limited time (\textit{challenge period}) to construct the proof. During this limited time, there is a negligible chance that the malicious prover could achieve a higher \textit{overall difficulty} than the honest prover, because, the honest mining power is more than the adversary mining power in the network.
 
 The detailed analysis of the protocol is described in Section~\ref{sec:protocol-analysis}. After running the above protocol, the verifier has the last finalized block header of the blockchain. In the following, we explain how the verifier confirms data from the history of the blockchain.
 
 

\begin{algorithm}
\caption{LightSync Protocol for Finding the Last Finalized Block Header}
\label{alg1}
\begin{algorithmic}[1]
\STATE The verifier creates a query transaction and sends it to the selected provers along with the challenge period.
\STATE The verifier starts a timer for the challenge period.
\STATE Each prover who receives the query transaction runs \textsc{CreateProof}($queryTransaction$) (Procedure \ref{alg2}) and sends the created proof to the verifier.
\STATE For each received proof, the verifier runs \textsc{ValidateProof}($proof$) (Procedure \ref{alg3}).
\STATE The verifier runs \textsc{OverallDifficulty}($proof, queryTransaction$) (Procedure \ref{alg4}) for the valid proofs.
\STATE The verifier chooses the valid proof with the highest overall difficulty as the winner proof.
\end{algorithmic}
\end{algorithm}

 \begin{algorithm}\floatname{algorithm}{Procedure}
\caption{\textsc{CreateProof}(\textit{queryTransaction})}
\label{alg2}
\begin{algorithmic}[1]
\STATE Start a timer for the challenge period.
\STATE $\delta$ $\leftarrow$ the communication delay between the verifier and the prover
\STATE Checks the validity of $queryTransaction$
\STATE Broadcast $queryTransaction$ in the local blockchain
\REPEAT
\STATE On the local longest chain $B[0:end]$, find the block B[q] which includes $queryTransaction$
\STATE $MerkleProof$ $\leftarrow$ proof of inclusion of $queryTransaction$ in the $B[q]$
\STATE $m \leftarrow len(B[q:end])$
\IF {$m \geq k+1$} 
\STATE $C[0:m-1] \leftarrow B[q:end]$
\ELSE
\STATE $C[0:k] \leftarrow B[q-k+m:end]$
\ENDIF
\UNTIL{$2 \delta$ seconds remaining from the challenge period}
\STATE $proof$ $\leftarrow$ $C$ $\cup$ $MerkleProof$
\STATE \textbf{return} $proof$
\end{algorithmic}
\end{algorithm}
 
\begin{algorithm}\floatname{algorithm}{Procedure}
\caption{\textsc{ValidateProof}(\textit{proof})}
\label{alg3}
\begin{algorithmic}[1]
\STATE Check the length of the chain in the $proof$ to be at least $k+1$.
\STATE Verify the correctness of Merkle inclusion proof for the query transaction.
\STATE Check that each block header of the chain refers correctly to its previous block header.
\STATE Calculate the PoW for each block header of the chain and validate it against the target difficulty.
\STATE \textbf{return} $true$ if all the above checks pass
\end{algorithmic}
\end{algorithm}

\begin{algorithm}\floatname{algorithm}{Procedure}
\caption{\textsc{OverallDifficulty}(\textit{proof,  queryTransaction})}
\label{alg4}
\begin{algorithmic}[1]
\STATE Find the block header $C[q]$ containing $queryTransaction$
\STATE $T[i]$ := target difficulty of block header $C[i]$ ($q \leq i \leq end$)
\STATE $overallDifficulty$ $\leftarrow$ $\sum_{i=q}^{end} T[i]^{-1}$
\STATE \textbf{return} $overallDifficulty$
\end{algorithmic}
\end{algorithm}

 \subsection{Finding the Last Valid MMR Root}\label{part2}
 After executing the first part of the \textsc{LightSync} protocol, a verifier has got a copy of at least $k+1$ block headers of the blockchain. If the construction of the block headers of the blockchain includes the MMR root of all previous block headers of the blockchain, then knowing the last $k+1$ block headers, the client has access to the MMR root in a finalized block header of the blockchain. The client can perform a transaction inclusion test for transactions related to the past blocks of the blockchain using that MMR root. This is the case for some blockchains like Beam \cite{beam} and Grin \cite{grin} that include MMR root in their block header structure.
 
 However, if a blockchain does not support MMR, there exist three possible approaches for including MMR roots in the structure and using them: Hard fork, Soft fork, or Velvet fork. In the case of a hard fork, miners are required to include an MMR root in the block headers of the blockchain. Alternatively, in a soft fork, the MMR root is added to the new blocks in a way to stay backward compatible with the old blocks. For example, the MMR root can be included in a predetermined transaction like the coinbase transaction. This way the non-upgraded miners will accept the new blocks whether these blocks are created following the new rule or not. Nonetheless, the upgraded miners will follow the new rule and only accept blocks including a valid MMR root. A soft fork needs the majority of the miners to upgrade to the new rule. After including the MMR root in the structure of the block headers of a blockchain using either of the above two approaches, it can be used immediately to verify data from the history of the blockchain by the verifier. For example, Zcash \cite{zcash} has added MMR root to its block header structure using a hard fork named Heartwood. 
 
 The last approach is to use a velvet fork to add MMR roots to the block headers. In velvet forks, there is no need for a specific number of miners to upgrade. The upgraded miners will add the MMR root to the block headers in a backward-compatible way (like soft forks it can be included in the coinbase transaction data), but, they continue to accept blocks created by non-upgraded miners. This way the set of accepted blocks stays unchanged, meaning that every non-upgraded or upgraded miner will accept a block, whether it includes an MMR root or not. In this case, since no one is checking the validity of the MMR field, some of the blocks may contain invalid MMR roots. If the verifier accepts an invalid MMR root, the malicious prover has deceived the verifier to accept a predicate that contradicts the honest chain. To avoid invalid MMR roots getting accepted by the light clients, we append a small buffer of data to the MMR root that enables a voting mechanism for them. Applying this change, the upgrade remains backward compatible. Using this mechanism, we make sure that the verifier can detect a valid MMR root.
 
 The buffer's length is $\alpha$. Each bit of the buffer points to a preceding block and votes for its MMR root validity. In a block header, the last bit of its voting buffer refers to the last upgraded block, the bit before that refers to the second last upgraded block, and so on. The miner sets a bit to $1$ (\textit{accept vote}) if he believes that the corresponding block includes a valid MMR root and sets it to $0$ (\textit{reject vote}) otherwise.
 We assume that the majority of the upgraded miners are honest. This way if the majority of miners in a long enough sub-chain have voted $1$ for an MMR root, the verifier can make sure that MMR root is valid and by adding new block headers to it, she can easily construct the MMR root of the last finalized block by herself.
 
 The second part of the proof is sent by the prover to the verifier in the case of a velvet fork, to provide her with the latest valid MMR root. The verifier starts the second part of the protocol by sending the last finalized block header and a parameter called $\beta$ to the prover. In response to the verifier's request, the prover sends a sub-chain of block headers that includes $\alpha + \beta$ upgraded block headers, as well as each block's coinbase data with proof of its inclusion in the block. We call the first $\beta$ block headers of the sub-chain \textit{candidates}. In Section~\ref{sec:security}, we discuss how $\beta$ effects the protocol security. Algorithm~\ref{alg5} shows the steps of the second part of the protocol.

 First, the verifier validates the received proof. She checks the previous hashes of the block headers to be correct and the coinbase data inclusions to be valid. She also checks the extended proof to include the last finalized block header from the previous part. After checking the validity of the proof, the verifier extracts candidates' block headers. Each of the candidates has $\alpha$ voters vote for them, as long as every upgraded block header votes for previous $\alpha$ upgraded block headers. An MMR root among the candidates that has been accepted by more than half of its voters is a valid MMR root. The verifier will accept that MMR root and use it to obtain the last valid MMR root of the blockchain by appending the block headers that exist on top of it to its tree. By choosing the right $\beta$, there exists a valid MMR root among the candidates with an overwhelming probability. A detailed analysis of this voting mechanism and how to determine the parameter $\beta$ is discussed in Section~\ref{sec:protocol-analysis}.
 
 
 \begin{algorithm}
\caption{LightSync Protocol for Finding the Last Valid MMR Root}
\label{alg5}
\begin{algorithmic}[1]
\STATE The verifier sends an MMR request along with the last finalized block header $B[f]$ to a prover.
\STATE The prover creates a proof containing two parts:
\begin{enumerate}
    \item The last $s+1$ block headers up to $B[f]$, which is $B[f-s:f]$, that contains $\alpha+\beta$ upgraded block headers from his local chain $B[0:f]$.
    \item The coinbase data for each block header along with its proof of inclusion in that block header \textsc{proof}($CoinbaseData$).
\end{enumerate}
\STATE The prover sends the proof to the verifier.
\STATE The verifier runs \textsc{FindLastMMR}($B[f-s:f]$, $CoinbaseData[0:s]$, \textsc{proof}($CoinbaseData[0:s]$)) to find the last valid MMR root.
\end{algorithmic}
\end{algorithm}

\begin{algorithm}\floatname{algorithm}{Procedure}
\caption{\textsc{FindLastMMR}($C[0:s]$, $CoinbaseData[0:s]$, \textsc{proof}($CoinbaseData[0:s]$))}
\label{alg6}
\begin{algorithmic}[1]
\STATE Check that $C[s]$ is equal to $B[l]$.
\STATE Check that the number of upgraded block headers is equal to $\alpha+\beta$.
\STATE Check that each block header refers correctly to its previous block header.
\STATE Check the correctness of each coinbase data inclusion proof. 
\STATE $candidates[0:\beta-1]$ $\leftarrow$ first $\beta$ upgraded block headers of $C[0:s]$
\STATE $V[0:\beta-1]$ := the number of accept votes for each block header of $candidates[0:\beta-1]$
\STATE $validSet$ $\leftarrow$ the $candidates$'s block headers with more than $\lfloor\alpha/2\rfloor$ accept votes
\STATE $validRoot$ $\leftarrow$ the MMR root of the last block header of $validSet$
\STATE $v$ $\leftarrow$ the index of the block containing $validRoot$
\STATE $lastValidRoot$ $\leftarrow$ update the $validRoot$ by adding the block headers $C[v:s]$ to its tree
\STATE \textbf{return} $lastValidRoot$

\end{algorithmic}
\end{algorithm}

 \subsection{Staying Up-to-Date}\label{part3}
 In some applications, after running the \textsc{LightSync} protocol and getting the last finalized block header and the last valid MMR root of the blockchain, the verifier may want to receive the future block headers as soon as they are mined. In this case, the verifier can continue requesting block headers from the honest full node and check their validity and compliance with her current chain of block headers and add them to her local chain if the received data is valid. When a new block header gets finalized, the verifier can add it to the MMR root she maintains to achieve the last stable MMR root of the blockchain. Then, she can drop all the finalized block headers prior to it. In this way, she will always maintain a limited number of block headers; a sub-chain with length $k+1$ with some possible forks. However, she can perform any transaction inclusion or state verification by using the MMR root she maintains. 
 
 In other cases where the verifier doesn't need to learn about new block headers immediately as they are mined, she can repeat the \textsc{LightSync} protocol in long time intervals. In this way, instead of getting each block header in real-time, the client can update herself periodically. So, she will use fewer resources, but, will experience some delays.

 \section{Incentive Design}\label{sec:incentive-design}

 Full nodes maintain a copy of the whole blockchain and relay data in the network. In addition, they provide the verifiers with light client services which have some extra costs for them. The growing number of light clients makes the services' cost considerable for full nodes. Without a proper incentive mechanism, there is a lack of motivation for users to participate in the network as full nodes. In this section, we introduce \textsc{LightSync} incentive mechanism, which proposes a way to cover the costs of being a full node and prevents the light clients from free-riding. The previous light client solutions suffer from the lack of incentive in their protocols. With the increase of demand for light clients, there is a need to properly incentivize full nodes, so that they continue providing the verifiers with light client services.

 \textbf{Incentive Mechanism.}
 To start the protocol, the verifier sends a \textit{query transaction} to provers. This transaction is a Pay-To-Script-Hash (P2SH)\cite{p2sh}: a special type of transaction that transfers money from the sender of the transaction to the hash of a script which is determined by the sender. The script may require some conditions to be satisfied. Anyone who provides the script and satisfies the dictated conditions can spend the P2SH transaction.
 The verifier uses this kind of transaction to pay the prover for the service she gets. We call this fee the \textit{Service Fee} that compensates for the costs of the full nodes. The verifier also has to pay the transaction fee of the P2SH transaction.
 
 After executing the \textsc{LightSync} protocol, the verifier sends the script to the prover who has provided the final chosen proof. Using the script, the prover can spend the \textit{query transaction} and receive his fee. The service fee makes an incentive for honest provers to participate in \textsc{LightSync} and broadcast the \textit{query transaction} in the fastest time. If more honest nodes are incentivized to participate in \textsc{LightSync}, the chance of the user connecting to at least one honest prover gets higher. Also, the faster the honest prover propagates the \textit{query transaction} in the network, the sooner an honest miner includes it in a block.

 A summary of all the needed communication between the verifier and the provers is shown in Figure \ref{fig2}. After receiving the proofs, the verifier detects the honest prover and continues communicating with him to obtain the MMR and provide him with the solution for conditions of P2SH transaction. We call it the puzzle of P2SH transaction.

 The described mechanism has a challenge: the verifier might refuse to provide the prover with the script. In this case, although the prover has provided a service for the verifier, it hasn't received the service fee and the verifier has only paid a transaction fee.
 
 We assume that after executing the \textsc{LightSync} protocol, the communication between the verifier and the prover needs to be continued. Since the verifier is a light client, she wants to receive every new block header from the prover. If the verifier refuses to reveal the script, the prover will not provide her with new block headers anymore. In this case, the verifier needs to restart the \textsc{LightSync} protocol with another prover. This means that she needs to send a new \textit{query transaction} to the other provers. This costs an extra transaction fee for her. By setting the transaction fee greater than the service fee, it will demotivate the verifier to abandon the first honest prover. 

 \begin{figure}[!t]
  \includegraphics[scale=0.4]{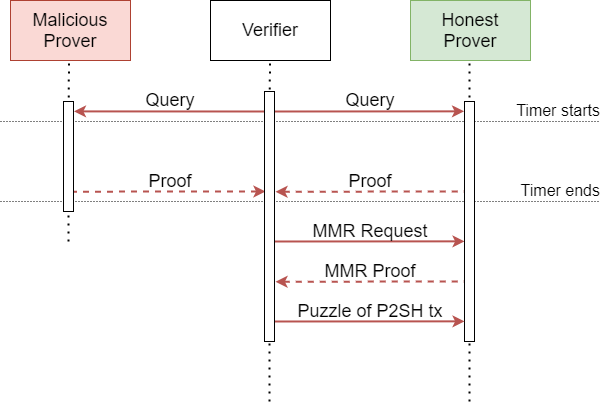}
  \centering
  \caption{Communication between the verifier and the provers}
  \label{fig2}
 \end{figure}
 
\section{Protocol Analysis} \label{sec:protocol-analysis}
 In this section, we perform a thorough analysis of the complexity and security of the \textsc{LightSync} protocol. We compare \textsc{LightSync} to the existing solutions in terms of needed resources and we show that it is secure.

\subsection{Complexity Analysis}\label{sec:complexity}

The main goal of a light client protocol is to enable light clients to verify the inclusion of some data in the blockchain using the least possible storage and computation resources.

In the first part of the protocol (\ref{part1}), each prover has a limited time (\textit{challenge period}) to provide the verifier with a proof. In \textsc{LightSync}, the \textit{challenge period} is determined by the desired level of security and is independent of the blockchain's length. The number of blocks mined in the honest chain during this period is on average the \textit{challenge period} divided by the block interval time of the blockchain, which stays constant. If MMR root is included in the block structure using a hard fork or a soft fork, after running the first part of the protocol, no more data is required from the prover. So, the verifier validates her predicate on the honest chain using constant resources.

If MMR structure has been added to the blockchain using velvet fork, after running the first part of the protocol, more data is required to validate the predicate. The winner prover sends a sub-chain that includes $\alpha+\beta$ upgraded block headers to the verifier. Assuming that the ratio of the total mining power to the upgraded mining power is $l$, in a long run, $1/l$ of total mined blocks of the honest chain are upgraded. So, on average, the winner prover needs to send $(\alpha+\beta)\times l$ block headers to the verifier to provide $\alpha+\beta$ upgraded block headers. Since $\alpha$, $\beta$, and $l$ do not depend on the length of the blockchain, the needed resources to perform the second part of the protocol (\ref{part2}) stays constant. 

In conclusion, no matter what the size of the blockchain is, the verifier can perform data inclusion verification using constant storage and computation resources. In Table~\ref{tab_comp}, we make a comparison between existing solutions and \textsc{LightSync}. The verifier's complexity in the table refers to the needed storage and computation resources for the verifier.

\begin{table}[!t]
\caption{Comparison of different light client protocols}
\label{tab_comp}
\centering
\begin{tabular}{|c|c|c|c|}
\hline
\textbf{Protocol} & \textbf{Protocol Complexity for Verifier} & \textbf{Added Structure}  \\ 
\hline
SPV & $O(n)$ & -\\ 
\hline
NIPoPoW  & $O(\polylog(n))$ & Interlink \\ 
\hline
FlyClient & $O(\polylog(n))$ & MMR \\ 
\hline
LightSync & $O(1)$ & MMR\\ 
\hline
\end{tabular}
\end{table}

\subsection{Security Analysis}\label{sec:security}
As we described in the prover and verifier model, the verifier sends queries to two provers where one of them is honest and the other one is malicious. Each prover sends a proof to the verifier. After executing the protocol, the verifier decides on one of the proofs as the winner proof. The \textsc{LightSync} protocol is secure if the verifier decides on the honest prover's proof with an overwhelming probability. In what follows, we explain the security conditions of the protocol and prove them.

We assume that the majority of the underlying blockchain miners are honest which is a necessary condition for Proof-of-Work blockchains. We assume the Poisson process model with rate $\lambda$ for the block mining process. 
In the case where MMR structure has been added to the blockchain using velvet fork (the approach described in \ref{part2}), we assume that the majority of the upgraded miners are honest.

Regarding the difficulty of the blocks, we assume that in the blocks that are sent as the first part of the proof, the difficulty of the blocks would change at most one time. This is a realistic assumption as long as the difficulty adjustment in a blockchain gets done in longer time intervals than the \textit{challenge period}.

The malicious prover wants to deceive the verifier and make her confirm some predicate that contradicts the honest chain data. The verifier validates her predicate using the last valid MMR root. So, the malicious prover should be able to convince the verifier to accept an invalid MMR root to succeed in deceiving the verifier. In the following, we analyze both cases where all the blocks include an MMR root (case I) and where the MMR root is added to some of them using a velvet fork (case II). The former is the case where the blockchain includes the MMR root in its structure (and so there is no need to do a fork), also, it is the case of using a hard fork or soft fork to add MMR to the structure. 

\textbf{Case I.} In this case where all blocks include an MMR root, the malicious prover should send an invalid block as a finalized block to be able to mislead the verifier, because, the last valid MMR root is included in the last finalized block header. If the malicious prover wants his invalid block to be the winner proof for the verifier, he has to include some confirmation blocks containing higher \textit{overall difficulty} than the other prover. Based on Theorem~\ref{theorem1}, the probability of this scenario is negligible under our model assumptions.

The probability of the query transaction being included in the first block mined after the \textit{challenge period} starts is determined by its transaction fee. The higher the transaction fee is, the higher will be the chance for it to be included in the blockchain sooner. We assume the transaction fee to be high enough so that the query transaction gets included in the first mined block after the challenge period starts. If the transaction fee is not high enough, the transaction will be included in later block headers of the honest chain, however, the adversary mining power will include it in the first block of its local chain anyways. Hence, there will be a higher chance for the malicious proof to have more overall difficulty and win.

The proof includes $m \geq k+1$ ($m' \geq k+1$) block headers, for the honest (malicious) prover. We denote this sub-chain  by $C[0:m-1]$ ($C'[0:m'-1]$). The \textit{query transaction} is included in one of the block headers of the proof, suppose the block header to be $C[q]$ ($C'[q']$). The verifier calculates the \textit{overall difficulty} of $C[q:end]$ and $C'[q':end]$. The proof with the higher \textit{overall difficulty} will be chosen as the \textit{winner proof}.

The target difficulty for each of these two sub-chains can change at most one time. We denote the number of block headers using the same target difficulty $T_1$ ($T'_1$) which are constructed in time $t_1$ ($t'_1$) starting from the beginning of the \textit{challenge period} by random variable $N_1$ ($N'_1$). In the remaining time $t_2$ ($t'_2$), blocks are constructed with target difficulty $T_2$ ($T'_2$). We denote the number of blocks constructed during this time with random variable $N_2$ ($N'_2$). Real numbers $t_1$, $t_2$, $t_1'$, and $t_2'$ are non-negative where $t_1+t_2$ and $t_1'+t_2'$ are equal to the \textit{challenge period} ($t$). Clearly, discrete random variables $N_1$, $N_2$, $N_1'$, and $N_2'$ have Poisson distributions.

We suppose the block arriving rate of the blockchain to be $\lambda_1$ blocks per second from the start of the \textit{challenge period} until $t_1$ seconds. Since the \textit{challenge period} is much less than the difficulty adjustment time interval, we can assume that the whole mining power of the network is constant during this time. Furthermore, with changing the target difficulty of the blockchain, the rate of arriving blocks will change to $\lambda_2$ after $t_1$ seconds. Parameters $\lambda_1'$ and $\lambda_2'$ are defined similarly for the adversary mining power. All the blocks of the adversary chain should meet their target difficulties, otherwise, the verifier would not accept their proof as a valid proof. However, the adversary mining power will not necessarily use the target difficulties of the honest chain. They may use other target difficulties in order to maximize their chance of being selected as the winner proof. Based on our model assumptions, the arrival rate of new blocks in a chain is proportional to the target difficulty for constant active mining power. 

\begin{thm}\label{theorem1} The verifier sends a \textit{query transaction} (\textit{tx}) to two provers where exactly one of them is honest. The honest prover running the \textsc{LightSync} protocol provides \textit{proof}, and the malicious prover provides \textit{proof'}. The probability 
\begin{equation}\label{ineq_1}
\begin{aligned}
\Pr\{\textsc{overallDifficulty}(proof, tx)\\ > \textsc{overallDifficulty}(proof', tx)\} 
\end{aligned}
\end{equation}
is negligible.
\end{thm}

\begin{proof}
Using the notations we defined in this section, we know that:
\begin{equation}
\begin{aligned}
\textsc{overallDifficulty}(proof, tx)
=\frac{N_1}{T_1}+\frac{N_2}{T_2}
\end{aligned}
\end{equation}
and
\begin{equation}
\begin{aligned}
\textsc{overallDifficulty}(proof', tx)
 =\frac{N_1'}{T_1'}+\frac{N_2'}{T_2'}.
\end{aligned}
\end{equation}

So, the probability (\ref{ineq_1}) is equal to
\begin{equation}\label{ineq_2}
\begin{aligned}
\Pr\{N_1'>N_1\frac{T_1'}{T_1}+N_2\frac{T_1'}{T_2}-N_2'\frac{T_1'}{T_2'}\}\\
=\sum_{n_1=0}^\infty\sum_{n_2=0}^\infty\sum_{n_2'=0}^\infty \big(\Pr\{N_1'>n_1\frac{T_1'}{T_1}+n_2\frac{T_1'}{T_2}-n_2'\frac{T_1'}{T_2'}\}\\
\times \Pr\{N_1=n_1\}\Pr\{N_2=n_2\}\Pr\{N_2'=n_2'\}\big).
\end{aligned}
\end{equation}

Due to symmetry and $N_1, N_2, N_1'$, and $N_2'$ being independent random variables, we can assume that $T_1\leq T_2$ and $T_1'\leq T_2'$ without loss of generality. In other words, the scenarios in which a prover has used a higher target difficulty at first and then decreased it, can be matched with scenarios in which the time intervals for each target difficulty remain the same but the lower target difficulty comes first. As long as the number of mined blocks, time intervals, and target difficulties remain the same and only their order of occurrence changes, the probability of the malicious prover winning remains the same. 

Now using Lemma~\ref{lemma1} and the fact that $N_1, N_2$, and $N_2'$ have Poisson distributions, we can show that the Probability~\eqref{ineq_2} is less than or equal to

\begin{equation}\label{expression}
\begin{aligned}
\sum_{n_1=0}^\infty\big(\frac{{\bigg(\frac{\lambda_1t_1}{e^{m\frac{T_1'}{T_1}}}\bigg)}^{n_1}}{n_1!}\big)
\times\sum_{n_2=0}^\infty\big(\frac{{\bigg(\frac{\lambda_2t_2}{e^{m\frac{T_1'}{T_2}}}\bigg)}^{n_2}}{n_2!}\big)\times e^{\lambda_1't_1'(e^m-1)}\\
\times\sum_{n_2'=0}^\infty\big(\frac{{\bigg(\frac{\lambda_2't_2'}{e^{-m\frac{T_1'}{T_2'}}}\bigg)}^{n_2'}}{n_2'!}\big)
\times e^{-\lambda_1t_1}\times e^{-\lambda_2t_2}
\times e^{-\lambda_2't_2'}
\end{aligned}
\end{equation}
for any $m>0$.

By using Taylor series of an exponential function and substituting $t_2$ and $t_2'$ with $t-t_1$ and $t-t_1'$ respectively, and letting $m=m_0=(\ln{\frac{T_1'\lambda_1}{T_1\lambda_1'}})/(1+\frac{T_1'}{T_1})$, it gets proved that the expression (\ref{expression}) is equal to $T(m_0,t,t_1,t_1')$ where function $T(m,t,t_1,t_1')$ is defined in Equation (\ref{define_T}). Given the definition of target difficulty and the honest mining power being more than the adversary mining power, we get $\frac{\lambda_1}{T_1}=\frac{\lambda_2}{T_2}>\frac{\lambda_1'}{T_1'}=\frac{\lambda_2'}{T_2'}$. Using Lemma~\ref{lemma3}, we conclude that $T(m_0,t,t_1,t_1')$ exponentially decreases by increasing $t$, therefore, by choosing a proper $t$ (which is the \textit{challenge period}) the result follows.
\end{proof}

\textbf{Case II.}
In this case, the probability of the verifier getting an invalid finalized block in the first part of the protocol is similar to the previous case which we proved is negligible. In the second part of the protocol, where the verifier aims to find the last valid MMR root, for the protocol to be secure, we need to prove two important statements:
\begin{itemize}
    \item The \textit{candidates} contain a block header including a valid MMR root with an overwhelming probability.
    \item After running the second part of the protocol, the verifier will decide on a block header among the \textit{candidates} including a valid MMR root with an overwhelming probability.
\end{itemize}

Given that a block mined by an upgraded honest node includes a valid MMR root, Theorem \ref{theorem2} results in the first statement. Based on Theorem \ref{theorem3}, the probability of an MMR root that is inconsistent with the honest chain to be chosen as valid MMR root by the verifier is negligible when $\alpha$ is chosen properly. It can be easily deduced that the probability of a valid MMR root getting rejected by the majority of the voters is the same and is negligible. Therefore, the second statement is concluded. To give an example for a proper $\alpha$, let's assume a blockchain's network where $M_a=\frac{1}{3}$ and let $\alpha=80$. Now, based on Equation (\ref{eq_s}), the probability of a wrong MMR root being chosen by the verifier is less than $0.01$.

Suppose the proportion of the upgraded honest mining power to be $M_h$ and the adversary mining power to be $M_a$. We have $M_h+M_a=1$. We assume that $M_h>M_a$.

\begin{thm}\label{theorem2}
Consider the candidates' block headers. The probability of none of them having been mined by an honest node is negligible by choosing the right $\beta$.
\end{thm}

\begin{proof}
The \textit{candidates} are $\beta$ block headers where each of them could have been mined by the honest or the adversary mining power. The probability of a block being mined by an honest miner is $M_h$. We know that the probability of blocks being mined by honest nodes is independent of each other. Therefore, the probability of the \textit{candidates} not including any honest block headers is ${M_a}^\beta$ which is negligible for a large enough $\beta$.
\end{proof}

To illustrate this probability, let us consider an example here for a blockchain network where $M_a=\frac{1}{3}$. Letting $\beta=7$, the above probability will be less than $0.0005$.

\begin{thm}\label{theorem3}
For $\alpha \rightarrow \infty$, the probability of an inconsistent MMR root ($R_{ic}$) getting more than $\lceil\frac{\alpha}{2}\rceil$ accept votes goes to zero.
\end{thm}
\begin{proof}
As $R_{ic}$ is not consistent with the honest chain, the honest miners cast reject votes for it. At most, all the malicious miners cast accept vote for $R_{ic}$. The number of malicious miners in the next $\alpha$ upgraded blocks is the binomial random variable $X$ with parameters $p = M_a$ and $n = \alpha$. The probability of $X$ being greater than or equal to $\lceil\frac{\alpha}{2}\rceil$ is:
\begin{equation} \label{eq_s}
\begin{aligned}
\Pr(X \geq \lceil\frac{\alpha}{2}\rceil) = \sum_{i=\lceil\frac{\alpha}{2}\rceil}^{\alpha}{{\alpha}\choose i}(M_a^i)(M_h^{\alpha -i})\\
=M_a^{\lceil\frac{\alpha}{2}\rceil}\sum_{i=\lceil\frac{\alpha}{2}\rceil}^{\alpha}{{\alpha}\choose i}(M_a^{i-\lceil\frac{\alpha}{2}\rceil})(M_h^{\alpha-i})\\
< M_a^{\lceil\frac{\alpha}{2}\rceil}\sum_{i=\lceil\frac{\alpha}{2}\rceil}^{\alpha}{{\alpha}\choose i}(M_h^{\lfloor\frac{\alpha}{2}\rfloor})
\leq(M_aM_h)^{\frac{\alpha-1}{2}}\times 2^{\alpha-1}
\end{aligned}
\end{equation}

We know that
$M_a+M_h=1 \wedge M_a< M_h \Rightarrow M_aM_h<\frac{1}{4} \Rightarrow M_aM_h=2^{-(2+s)}$,
where $s>0$. Therefore, $\Pr(X \geq \lceil\frac{\alpha}{2}\rceil)<2^{-s\frac{(\alpha-1)}{2}}$,
where $s=-\log_2(M_aM_h)-2>0$. So, $\lim_{\alpha\rightarrow\infty} \Pr(X \geq \lceil\frac{\alpha}{2}\rceil)=0$.
\end{proof}

\section{Conclusion}

In this paper, we present \textsc{LightSync}, a low-cost light client protocol that is suitable for Proof of Work blockchains. Leveraging \textsc{LightSync}, a light client can verify the inclusion of data in the blockchain using constant computation and storage resources no matter what the length of the blockchain is. \textsc{LightSync} also introduces an incentive mechanism to compensate full nodes that provide data for the light clients. We describe the details of applying \textsc{LightSync} and perform a thorough security analysis of it. \textsc{LightSync} can be used for lightweight devices (e.g., mobile phones and IoT devices) that intend to connect to blockchains and verify data against them. Another important application of such a protocol is in cross-chain communication, in which one blockchain wants to verify the correctness of data against the other one. For future work, we are going to extend the \textsc{LightSync} protocol to support blockchains using other consensus mechanisms (e.g. Proof of Stake). Also, we will design a cross-chain bridge using the \textsc{LightSync}.

\appendix

In this appendix, we present several technical lemmas used in Section~\ref{sec:protocol-analysis}. The proofs of these lemmas are omitted due to the space constraint.

\section{Technical Lemmas}
\begin{lem}\label{lemma1} Suppose that $N_1'$ has a Poisson distribution with rate $\lambda_1't_1'$. We show that

\begin{equation}
\begin{aligned}
\Pr\{N_1'>n_1\frac{T_1'}{T_1}+n_2\frac{T_1'}{T_2}-n_2'\frac{T_1'}{T_2'}\}\\
\leq \frac{e^{(e^m-1)\lambda_1't_1'}}{(e^{mn_1\frac{T_1'}{T_1}})(e^{mn_2\frac{T_1'}{T_2}})(e^{-mn_2'\frac{T_1'}{T_2'}})}
\end{aligned}
\end{equation}

where $m,T_1,T_2,T_1',T_2'>0$ and $n_1,n_1',n_2,n_2'\geq0$.
\end{lem}

\begin{proof} For some $m>0$, using the Chernoff bound we have:
\begin{equation}
\begin{aligned}
Pr\{N_1'>n_1\frac{T_1'}{T_1}+n_2\frac{T_1'}{T_2}-n_2'\frac{T_1'}{T_2'}\}\\
\leq \frac{\EX[e^{mN_1'}]}{e^{mn_1\frac{T_1'}{T_1}+mn_2\frac{T_1'}{T_2}-mn_2'\frac{T_1'}{T_2'}}}
\end{aligned}
\end{equation}
Now, notice
\begin{equation}
\begin{aligned}
\EX[e^{mN_1'}]=\sum_{n_1'=0}^{\infty}e^{mn_1'}Pr\{N_1'=n_1'\}=e^{(e^m-1)\lambda_1't_1'}
\end{aligned}
\end{equation}
and the result follows.
\end{proof}

\begin{lem}\label{lemma2} Suppose that $T_1,T_1',\lambda_1,\lambda_1'>0$ and $\frac{\lambda_1}{T_1}>\frac{\lambda_1'}{T_1'}$. For function
\begin{equation}
\begin{aligned}
f(t)
=(e^{m_0}-1)\lambda_1't-\lambda_1t+\lambda_1e^{-m_0\frac{T_1'}{T_1}}t
\end{aligned}
\end{equation}
where $m_0=\frac{\ln{\frac{T_1'\lambda_1}{T_1\lambda_1'}}}{1+\frac{T_1'}{T_1}}$, we have $\lim_{t\rightarrow\infty}{e^{f(t)}}=0$.
\end{lem}


\begin{proof}
It is sufficient to show
\begin{equation}
\begin{aligned}
\lim_{t\rightarrow\infty}{f(t)}=-\infty
\end{aligned}
\end{equation}
which is equivalent to showing that
\begin{equation}\label{ineq_g}
\begin{aligned}
\frac{f(t)}{t}<0
\end{aligned}
\end{equation}
We define
\begin{equation}
\begin{aligned}
g(m)=(e^{m}-1)\lambda_1'-\lambda_1+\lambda_1e^{-m\frac{T_1'}{T_1}}
\end{aligned}
\end{equation}
so, inequality (\ref{ineq_g}) is equivalent to $g(m_0)<0$.

Considering $\frac{\lambda_1}{T_1}>\frac{\lambda_1'}{T_1'}$, we get $g'(m)|_{m=0}<0$. We know $g(0)=0$, also, as the given $m_0$ is the only value for $m$ satisfying $g'(m)=0$, we conclude $g(m_0)<0$.
\end{proof}

\begin{lem}\label{lemma3}
Suppose that $T_2\geq T_1$, $T'_2\geq T'_1$, $t_1 + t_2 = t'_1 + t'_2$, $\frac{\lambda_1}{T_1} = \frac{\lambda_2}{T_2}>\frac{\lambda'_1}{T'_1} = \frac{\lambda'_2}{T'_2}$. Let $m_0=\frac{\ln{\frac{T_1'\lambda_1}{T_1\lambda_1'}}}{1+\frac{T_1'}{T_1}}$. Then, we have
\begin{equation}
\begin{aligned}
\lim_{t\rightarrow\infty}{T(m_0,t,t_1,t'_1)}=0
\end{aligned}
\end{equation}
where  
\begin{equation}\label{define_T}
\begin{aligned}
    T(m,t,t_1,t'_1) =  
    e^{{t_1}(-\lambda_1 + \lambda_2 + \lambda_1e^{-m\frac{T'_1}{T_1}}- \lambda_2e^{-m\frac{T'_1}{T_2}})}\\
    \times e^{t'_1(\lambda'_1e^m - \lambda'_2e^{m\frac{T'_1}{T'_2}} - \lambda'_1+ \lambda'_2)}\\
    \times e^{t(\lambda'_2e^{m\frac{T'_1}{T'_2}} -\lambda_2 - \lambda'_2 + \lambda_2e^{-m\frac{T'_1}{T_2}})}.
\end{aligned}
\end{equation}
\end{lem}

\begin{proof}
Let's define the below functions
\begin{equation}
\begin{aligned}
f_1(m) = -\lambda_1 + \lambda_2 + \lambda_1e^{-m\frac{T'_1}{T_1}}- \lambda_2e^{-m\frac{T'_1}{T_2}}
\end{aligned}
\end{equation}
\begin{equation}
\begin{aligned}
    \\f_2(m) = \lambda'_1e^m - \lambda'_2e^{m\frac{T'_1}{T'_2}} - \lambda'_1+ \lambda'_2
\end{aligned}
\end{equation}

The derivative of $f_1$ is equal to
\begin{equation}
\begin{aligned}
f_1'(m) = -\lambda_1\frac{T'_1}{T_1}e^{-m\frac{T'_1}{T_1}}
+\lambda_2\frac{T'_1}{T_2}e^{-m\frac{T'_1}{T_2}} 
\\=\frac{\lambda_1}{T_1}T'_1(e^{-m\frac{T'_1}{T_2}} - e^{-m\frac{T'_1}{T_1}})
\end{aligned}
\end{equation}

As $T_2 \geq T_1$, we get $f_1'(m) \geq 0$ for $m>0$ and since $f_1(0) = 0$, we conclude $f_1(m) \geq 0$ for $m>0$.

The derivative of $f_2$ is equal to
\begin{equation}
\begin{aligned}
    f_2'(m) = \lambda'_1e^m -\lambda'_2\frac{T'_1}{T'_2}e^{m\frac{T'_1}{T'_2}} = \lambda'_1(e^m-e^{m\frac{T'_1}{T'_2}})
\end{aligned}
\end{equation}

As $T'_2 \geq T'_1$, we get $f_2'(m) \geq 0$ for $m>0$ and since $f_2(0) = 0$, we conclude $f_2(m) \geq 0$ for $m>0$.

Given that $t \geq t_1$, $t \geq t'_1$, $f_1(m) \geq 0$, and $f_2(m) \geq 0$ for $m>0$, we get
\begin{equation}
\begin{aligned}
    T(m,t,t_1,t'_1) \leq e^{t(f_1(m))}e^{t(f_2(m))}e^{t(\lambda'_2e^{m\frac{T'_1}{T'_2}} -\lambda_2 - \lambda'_2 + \lambda_2e^{-m\frac{T'_1}{T_2}})}
    \\= e^{t((e^{m} -1)\lambda'_1-\lambda_1 + \lambda_1e^{-m\frac{T'_1}{T_1}})}
\end{aligned}
\end{equation}

 Based on Lemma \ref{lemma2}, for $m=m_0$ the above expression goes to zero when $t \rightarrow \infty$.
\end{proof}


\bibliographystyle{IEEEtran}
\bibliography{references}

\begin{thebibliography}{10}
\providecommand{\url}[1]{#1}
\csname url@samestyle\endcsname
\providecommand{\newblock}{\relax}
\providecommand{\bibinfo}[2]{#2}
\providecommand{\BIBentrySTDinterwordspacing}{\spaceskip=0pt\relax}
\providecommand{\BIBentryALTinterwordstretchfactor}{4}
\providecommand{\BIBentryALTinterwordspacing}{\spaceskip=\fontdimen2\font plus
\BIBentryALTinterwordstretchfactor\fontdimen3\font minus
  \fontdimen4\font\relax}
\providecommand{\BIBforeignlanguage}[2]{{%
\expandafter\ifx\csname l@#1\endcsname\relax
\typeout{** WARNING: IEEEtran.bst: No hyphenation pattern has been}%
\typeout{** loaded for the language `#1'. Using the pattern for}%
\typeout{** the default language instead.}%
\else
\language=\csname l@#1\endcsname
\fi
#2}}
\providecommand{\BIBdecl}{\relax}
\BIBdecl

\bibitem{Ethereum}
\BIBentryALTinterwordspacing
V.~Buterin, ``Ethereum white paper: A next generation smart contract \&
  decentralized application platform,'' 2014. [Online]. Available:
  \url{https://cryptorating.eu/whitepapers/Ethereum/Ethereum\_white\_paper.pdf}
\BIBentrySTDinterwordspacing

\bibitem{ethsize}
\BIBentryALTinterwordspacing
``Ethereum chain full sync data size,'' (Accessed on July 2021). [Online].
  Available:
  \url{https://ycharts.com/indicators/ethereum\_chain\_full\_sync\_data\_size}
\BIBentrySTDinterwordspacing

\bibitem{interop}
\BIBentryALTinterwordspacing
V.~Buterin, ``Chain interoperability,'' 2016. [Online]. Available:
  \url{https://www.r3.com/wp-content/uploads/2017/06/chain\_interoperability\_r3.pdf}
\BIBentrySTDinterwordspacing

\bibitem{Zamyatin2019SoKCA}
A.~Zamyatin, M.~Al-Bassam, D.~Zindros, E.~Kokoris-Kogias, P.~Moreno-Sanchez,
  A.~Kiayias, and W.~Knottenbelt, ``Sok: Communication across distributed
  ledgers,'' \emph{IACR Cryptol. ePrint Arch.}, vol. 2019, p. 1128, 2019.

\bibitem{belchior2021survey}
R.~Belchior, A.~Vasconcelos, S.~Guerreiro, and M.~Correia, ``A survey on
  blockchain interoperability: Past, present, and future trends,'' 2021.

\bibitem{btcrelay}
\BIBentryALTinterwordspacing
``Btc relay, a bridge between the bitcoin blockchain \& ethereum smart
  contracts,'' (Accessed on July 2021). [Online]. Available:
  \url{http://btcrelay.org/}
\BIBentrySTDinterwordspacing

\bibitem{rainbow}
\BIBentryALTinterwordspacing
``Eth-near rainbow bridge,'' (Accessed on July 2021). [Online]. Available:
  \url{https://near.org/blog/eth-near-rainbow-bridge/}
\BIBentrySTDinterwordspacing

\bibitem{gavzi2019proof}
P.~Ga{\v{z}}i, A.~Kiayias, and D.~Zindros, ``Proof-of-stake sidechains,'' in
  \emph{2019 IEEE Symposium on Security and Privacy (SP)}.\hskip 1em plus 0.5em
  minus 0.4em\relax IEEE, 2019, pp. 139--156.

\bibitem{Bitcoin}
\BIBentryALTinterwordspacing
S.~Nakamoto, ``Bitcoin: A peer-to-peer electronic cash system,'' 2009.
  [Online]. Available: \url{http://www.bitcoin.org/bitcoin.pdf}
\BIBentrySTDinterwordspacing

\bibitem{NiPoPoWs}
A.~Kiayias, A.~Miller, and D.~Zindros, ``Non-interactive proofs of
  proof-of-work,'' in \emph{Financial Cryptography and Data Security},
  J.~Bonneau and N.~Heninger, Eds.\hskip 1em plus 0.5em minus 0.4em\relax Cham:
  Springer International Publishing, 2020, pp. 505--522.

\bibitem{FlyClient}
B.~Bünz, L.~Kiffer, L.~Luu, and M.~Zamani, ``Flyclient: Super-light clients
  for cryptocurrencies,'' in \emph{2020 IEEE Symposium on Security and Privacy
  (SP)}, 2020, pp. 928--946.

\bibitem{velvetfork}
A.~Zamyatin, N.~Stifter, A.~Judmayer, P.~Schindler, E.~Weippl, and W.~J.
  Knottenbelt, ``A wild velvet fork appears! inclusive blockchain protocol
  changes in practice,'' in \emph{Financial Cryptography and Data Security},
  A.~Zohar, I.~Eyal, V.~Teague, J.~Clark, A.~Bracciali, F.~Pintore, and
  M.~Sala, Eds.\hskip 1em plus 0.5em minus 0.4em\relax Berlin, Heidelberg:
  Springer Berlin Heidelberg, 2019, pp. 31--42.

\bibitem{mmr}
\BIBentryALTinterwordspacing
P.~Todd and D.~Zindros, ``Merkle mountain ranges,'' 2018 (Accessed on July
  2021). [Online]. Available:
  \url{https://github.com/opentimestamps/opentimestamps-server/blob/master/doc/merkle-mountain-range.md}
\BIBentrySTDinterwordspacing

\bibitem{cryptoeprint:2021:782}
T.~Nemoz and A.~Zamyatin, ``On the deployment of flyclient as a velvet fork:
  chain-sewing attacks and countermeasures,'' Cryptology ePrint Archive, Report
  2021/782, 2021, \url{https://eprint.iacr.org/2021/782}.

\bibitem{security}
S.~Paavolainen and C.~Carr, ``Security properties of light clients on the
  ethereum blockchain,'' \emph{IEEE Access}, vol.~8, pp. 124\,339--124\,358,
  2020.

\bibitem{10.1007/978-3-319-63688-7_10}
J.~Garay, A.~Kiayias, and N.~Leonardos, ``The bitcoin backbone protocol with
  chains of variable difficulty,'' in \emph{Advances in Cryptology -- CRYPTO
  2017}, J.~Katz and H.~Shacham, Eds.\hskip 1em plus 0.5em minus 0.4em\relax
  Cham: Springer International Publishing, 2017, pp. 291--323.

\bibitem{merkle}
R.~C. Merkle, ``A digital signature based on a conventional encryption
  function,'' in \emph{A Conference on the Theory and Applications of
  Cryptographic Techniques on Advances in Cryptology}, ser. CRYPTO '87.\hskip
  1em plus 0.5em minus 0.4em\relax Berlin, Heidelberg: Springer-Verlag, 1987,
  p. 369–378.

\bibitem{hardsoft}
\BIBentryALTinterwordspacing
V.~Buterin, ``Hard forks, soft forks, defaults and coercion,'' (Accessed on
  July 2021). [Online]. Available:
  \url{https://vitalik.ca/general/2017/03/14/forks\_and\_markets.html}
\BIBentrySTDinterwordspacing

\bibitem{beam}
\BIBentryALTinterwordspacing
``Beam description. comparison with classical mw,'' (Accessed on July 2021).
  [Online]. Available:
  \url{https://docs.beam.mw/BEAM\_Comparison\_with\_classical\_MW.pdf}
\BIBentrySTDinterwordspacing

\bibitem{grin}
\BIBentryALTinterwordspacing
``Merkle mountain ranges (mmr),'' (Accessed on July 2021). [Online]. Available:
  \url{https://docs.grin.mw/wiki/chain-state/merkle-mountain-range/}
\BIBentrySTDinterwordspacing

\bibitem{zcash}
\BIBentryALTinterwordspacing
Y.~Tong~Lai, J.~Prestwich, and G.~Konstantopoulos, ``Flyclient -
  consensus-layer changes,'' 2019. [Online]. Available:
  \url{https://zips.z.cash/zip-0221}
\BIBentrySTDinterwordspacing

\bibitem{p2sh}
\BIBentryALTinterwordspacing
``Pay to script hash,'' (Accessed on July 2021). [Online]. Available:
  \url{https://en.bitcoin.it/wiki/Pay\_to\_script\_hash}
\BIBentrySTDinterwordspacing

\end{thebibliography}

\end{document}